\documentclass[10pt]{article}
\usepackage{graphicx}
\usepackage{multirow}

\usepackage[T1]{fontenc}
\usepackage[sc]{mathpazo}
\usepackage{amsmath}
\usepackage{amssymb}
\usepackage{enumerate}
\usepackage{amsthm}
\usepackage{amsfonts,mathrsfs}

\usepackage{hyperref}
\hypersetup{pdfpagemode=UseNone}

\newcommand{\tinyspace}{\mspace{1mu}}

\newcommand{\abs}[1]{\left\lvert\tinyspace #1 \tinyspace\right\rvert}

\newcommand{\norm}[1]{\left\lVert\tinyspace #1 \tinyspace\right\rVert}

\newcommand{\setft}[1]{\mathrm{#1}}

\newcommand{\density}[1]{\setft{D}\left(#1\right)}

\newcommand{\supp}{{\operatorname{supp}}}

\def\real{\mathbb{R}}

\def\I{\mathbb{1}}

\newenvironment{mylist}[1]{\begin{list}{}{
    \setlength{\leftmargin}{#1}
    \setlength{\rightmargin}{0mm}
    \setlength{\labelsep}{2mm}
    \setlength{\labelwidth}{8mm}
    \setlength{\itemsep}{0mm}}}
    {\end{list}}

\def\ot{\otimes}

\newcommand{\out}[2]{| #1\rangle\langle #2 |}
\newcommand{\Inner}[2]{\left\langle #1 , #2\right\rangle}

\newcommand{\defeq}{\stackrel{\smash{\textnormal{\tiny def}}}{=}}

\newcommand{\pa}[1]{(#1)}
\newcommand{\Pa}[1]{\left(#1\right)}

\newcommand{\Br}[1]{\left[#1\right]}
\newcommand{\set}[1]{\{#1\}}
\newcommand{\Set}[1]{\left\{#1\right\}}

\newcommand{\ket}[1]{|#1\rangle}

\DeclareMathOperator{\trace}{Tr}
\newcommand{\ptr}[2]{\trace_{#1}\pa{#2}}
\newcommand{\Ptr}[2]{\trace_{#1}\Pa{#2}}

\newcommand{\Tr}[1]{\Ptr{}{#1}}

\def\cH{\mathcal{H}}
\def\cM{\mathcal{M}}

\def\rF{\mathrm{F}}\def\rH{\mathrm{H}}

\def\rS{\mathrm{S}}

\def\sD{\mathscr{D}}

\def\sM{\mathscr{M}}

\newtheorem{thrm}{Theorem}[section]

\newtheorem{prop}[thrm]{Proposition}
\newtheorem{cor}[thrm]{Corollary}
\theoremstyle{definition}

\newtheorem{remark}[thrm]{Remark}

\numberwithin{equation}{section}

\newcounter{questionnumber}

\begin{document}

\title{\Large A lower bound of quantum conditional mutual information}

\author{Lin Zhang$^1$\footnote{E-mail: godyalin@163.com;
linyz@zju.edu.cn}\ , Junde Wu$^2$\\
  {\small $^1$\it Institute of Mathematics, Hangzhou Dianzi University, Hangzhou 310018, PR~China}\\
  {\small $^2$\it Department of Mathematics, Zhejiang University, Hangzhou 310027, PR~China}
   }
\date{}
\maketitle
\maketitle \mbox{}\hrule\mbox\\
\begin{abstract}

In this paper, a lower bound of quantum conditional mutual
information is obtained by employing the Peierls-Bogoliubov
inequality and Golden Thompson inequality. Comparison with the
bounds obtained by other researchers indicates that our result is
independent of any measurements. It may give some new insights over
squashed entanglement and perturbations of Markov chain states.

\end{abstract}
\maketitle \mbox{}\hrule\mbox\\

\section{Introduction}

Let $\cH$ be a finite dimensional complex Hilbert space. A
\emph{quantum state} $\rho$ on $\cH$ is a positive semi-definite
operator of trace one, in particular, for each unit vector
$\ket{\psi} \in \cH$, the operator $\rho = \out{\psi}{\psi}$ is said
to be a \emph{pure state}. The set of all quantum states on $\cH$ is
denoted by $\density{\cH}$. For each quantum state
$\rho\in\density{\cH}$, its von Neumann entropy is defined by
$$
\rS(\rho) := - \Tr{\rho\log\rho}.
$$
The \emph{relative entropy} of two mixed states $\rho$ and $\sigma$
is defined by
$$
\rS(\rho||\sigma) := \left\{\begin{array}{ll}
                             \Tr{\rho(\log\rho -
\log\sigma)}, & \text{if}\ \supp(\rho) \subseteq
\supp(\sigma), \\
                             +\infty, & \text{otherwise}.
                           \end{array}
\right.
$$
A \emph{quantum channel} $\Phi$ on $\cH$ is a trace-preserving
completely positive linear map defined over the set $\density{\cH}$.
It follows that there exists linear operators $\set{K_\mu}_\mu$ on
$\cH$ such that $\sum_\mu K^\dagger_\mu K_\mu = \I$ and $\Phi =
\sum_\mu \mathrm{Ad}_{K_\mu}$, that is, for each quantum state
$\rho$, we have the Kraus representation
\begin{eqnarray*}
\Phi(\rho) = \sum_\mu K_\mu \rho K^\dagger_\mu.
\end{eqnarray*}

The celebrated strong subadditivity (SSA) inequality of quantum
entropy, proved by Lieb and Ruskai in \cite{Lieb1973},
\begin{eqnarray}\label{eq:SSA-1} \rS(\rho_{ABC}) +
\rS(\rho_B) \leqslant \rS(\rho_{AB}) + \rS(\rho_{BC}),
\end{eqnarray}
is a very powerful tool in quantum information theory. Recently, the
operator extension of SSA is obtained by Kim in \cite{Kim2012}.
Following the line of Kim, Ruskai gives a family of new operator
inequalities in \cite{Ruskai2012}.

Conditional mutual information, measuring the correlations of two
quantum systems relative to a third, is defined as follows: Given a
tripartite state $\rho_{ABC}$, it is defined by
\begin{eqnarray}
I(A:C|B)_\rho := \rS(\rho_{AB})+ \rS(\rho_{BC}) - \rS(\rho_{ABC}) -
\rS(\rho_B).
\end{eqnarray}
Clearly conditional mutual information is nonnegative by SSA.

Ruskai is the first one to discuss the equality condition of SSA, that is, $I(A:C|B)_\rho = 0$. By
analyzing the equality condition of Golden-Thompson inequality, she
obtained the following characterization \cite{Ruskai2002}:
\begin{eqnarray}
I(A:C|B)_\rho = 0 \Longleftrightarrow \log\rho_{ABC} + \log\rho_B =
\log\rho_{AB} + \log\rho_{BC}.
\end{eqnarray}
Throughout the present paper, we have suppressed implicit tensor products with the identity by conventions. For example, $\log\rho_{AB}$ means $(\log\rho_{AB})\ot\I_C$.

Later on, using the relative modular approach established by Araki,
Petz gave another characterization of the equality condition of SSA
\cite{Petz2003}:
\begin{eqnarray}
I(A:C|B)_\rho = 0 \Longleftrightarrow
\rho^{\mathrm{i}t}_{ABC}\rho^{-\mathrm{i}t}_{BC} =
\rho^{\mathrm{i}t}_{AB} \rho^{-\mathrm{i}t}_B\quad(\forall
t\in\real),
\end{eqnarray}
where $\mathrm{i} = \sqrt{-1}$ is the imaginary unit.

Hayden \emph{et al.} in \cite{Hayden2004} showed that $I(A:C|B)_\rho
=0$ if and only if the following conditions hold:
\begin{enumerate}[(i)]
\item $\cH_B = \bigoplus_k \cH_{b^L_k} \ot \cH_{b^R_k}$,
\item $\rho_{ABC} = \bigoplus_k  p_k \rho_{Ab^L_k} \ot \rho_{b^R_kC}$, where $\rho_{Ab^L_k}\in\density{\cH_A \ot \cH_{b^L_k}}, \rho_{b^R_kC} \in \density{\cH_{b^R_k} \ot\cH_C}$
for each index $k$; and $\set{p_k}$ is a probability distribution.
\end{enumerate}

In order to avoid computations already known to be difficult, such
as those of logarithmic and complex exponential powers of states,
Zhang \cite{Zhang2013} gave a new characterization of $I(A:C|B)_\rho
= 0$. To be specific, define
\begin{eqnarray*}
M &\defeq&
(\rho^{1/2}_{AB}\ot\I_C)(\I_A\ot\rho^{-1/2}_B\ot\I_C)(\I_A\ot\rho^{1/2}_{BC})\equiv
\rho^{1/2}_{AB}\rho^{-1/2}_B\rho^{1/2}_{BC}.
\end{eqnarray*}
Then the following conditions are equivalent:
\begin{enumerate}[(i)]
\item $I(A:C|B)_\rho = 0$;
\item $\rho_{ABC} = MM^\dagger = \rho^{1/2}_{AB}\rho^{-1/2}_B\rho_{BC}\rho^{-1/2}_B \rho^{1/2}_{AB}$;
\item $\rho_{ABC} = M^\dagger M = \rho^{1/2}_{BC}\rho^{-1/2}_B\rho_{AB}\rho^{-1/2}_B \rho^{1/2}_{BC}$.
\end{enumerate}

In \cite{Brandao2011}, Brand\~{a}o \emph{et al.} first obtained the
following lower bound for $I(A:C|B)_\rho$:
\begin{eqnarray}\label{eq:fernando}
I(A:C|B)_\rho \geqslant \frac18
\min_{\sigma_{AC}\in\mathbb{SEP}}\norm{\rho_{AC} -
\sigma_{AC}}^2_{1\text{-}\mathbb{LOCC}},
\end{eqnarray}
where $\mathbb{SEP}$ denotes the set of all separable states on the composite system of $A$ and $C$, and
$$
\norm{\rho_{AC} - \sigma_{AC}}_{1\text{-}\mathbb{LOCC}} \defeq
\sup_{\cM\in1\text{-}\mathbb{LOCC}}\norm{\cM(\rho_{AC}) -
\cM(\sigma_{AC})}_1.
$$
The $\mathbb{LOCC}$ means local operation and classical communication. $1\text{-}\mathbb{LOCC}$ means that the set of one-way $\mathbb{LOCC}$ measurements. Note that the inequality \eqref{eq:fernando} holds only for $1\text{-}\mathbb{LOCC}$ norm, i.e. \emph{one-way} LOCC norm instead of $\mathbb{LOCC}$ norm since there exists a counterexample that violates LOCC-norm. The specific explanation about this can be found in [{\em Erratum to: Faithful Squashed Entanglement},
\href{http://dx.doi.org/10.1007/s00220-012-1584-y}{\textbf{316},
287-288 (2012).}]

Based on this result, they cracked a \emph{long-standing} open
problem in quantum information theory. That is, the squashed
entanglement is \emph{faithful}. Later, Li and Winter in
\cite{Li2014} gave another approach to study the same problem and
improved the lower bound for $I(A:C|B)_\rho$:
\begin{eqnarray}\label{eq:kli}
I(A:C|B)_\rho \geqslant \frac12
\min_{\sigma_{AC}\in\mathbb{SEP}}\norm{\rho_{AC} -
\sigma_{AC}}^2_{1\text{-}\mathbb{LOCC}}.
\end{eqnarray}

Along with the above line, Ibinson \emph{et al.} in
\cite{Ibinson2008} studied the robustness of quantum Markov chains,
they employed the following famous
characterization of saturation of monotonicity inequality of
relative entropy, that is, let $\rho,\sigma\in\density{\cH}$, $\Phi$ be a
quantum channel defined over $\cH$. If
$\supp(\rho)\subseteq\supp(\sigma)$, then
\begin{eqnarray}
\rS(\rho||\sigma) = \rS(\Phi(\rho)||\Phi(\sigma))\quad\text{if and
only if}\quad \Phi^*_\sigma\circ\Phi(\rho) = \rho,
\end{eqnarray}
where $\Phi^*_\sigma =
\mathrm{Ad}_{\sigma^{1/2}}\circ\Phi^*\circ\mathrm{Ad}_{\Phi(\sigma)^{-1/2}}$, and $\Phi^*$ is the dual of $\Phi$ with respect to Hilbert-Schmidt inner product over operator space on $\cH$, i.e. $\Tr{\Phi^*(X)Y} = \Tr{X\Phi(Y)}$ for all operators $X,Y$ on $\cH$
\cite{Petz1988,Hiai2011}.

In order to establish our results, the following three inequality is useful:

\begin{prop}[Peierls-Bogoliubov Inequality,
\cite{Bebiano}]\label{prop:PB}
For two Hermitian matrices $H$ and $K$, it holds that
\begin{eqnarray}
\frac{\Tr{e^{H+K}}}{\Tr{e^{H}}}\geqslant
\exp\Br{\frac{\Tr{e^HK}}{\Tr{e^{H}}}}.
\end{eqnarray}
The equality occurs in the Peierls-Bogoliubov inequality if and only
if $K$ is a scalar matrix.
\end{prop}

\begin{prop}[Golden-Thompson Inequality,
\cite{Forrester}]\label{prop:GT}
For arbitrary Hermitian matrices
$A$ and $B$, one has
\begin{eqnarray}
\Tr{e^{A+B}}\leqslant \Tr{e^Ae^B}.
\end{eqnarray}
Moreover $\Tr{e^{A+B}} = \Tr{e^Ae^B}$ if and only if $[A,B]=0$, i.e.
$AB=BA$.
\end{prop}

\begin{prop}[Wasin-So Identity, \cite{So1,So2}]\label{prop:Wasin-so}
Let $A, B$ be two $n\times n$ Hermitian matrices. Then there exist
two $n\times n$ unitary matrices $U$ and $V$ such that
\begin{eqnarray}
\exp\Pa{\frac A2}\exp(B)\exp\Pa{\frac A2} = \exp\Pa{UAU^\dagger +
VBV^\dagger}.
\end{eqnarray}
\end{prop}

In this paper, from the observations made by Carlen and Lieb in \cite{Carlen}, a lower bound of quantum conditional mutual
information $I(A:C|B)_\rho$ is obtained by employing Peierls-Bobogliubov inequality
and Golden Thompson inequality in section 3. Comparison with the bounds obtained by
Brand\~{a}o \emph{et al}, and Li and Winter, respectively, indicates that our result is independent of any measurements. This result maybe gives some new insights over \emph{squashed entanglement} and perturbations of
\emph{Markov chain states}.

\section{Main results}
\begin{thrm}\label{th:newbound}
For a tripartite state $\rho_{ABC}$, we have
\begin{eqnarray}\label{eq:lower-bound}
I(A:C|B)_\rho\geqslant \norm{\sqrt{\rho_{ABC}} -
\sqrt{\exp(\log\rho_{AB} - \log\rho_B+\log\rho_{BC})}}^2_2.
\end{eqnarray}
In particular, $I(A:C|B)_\rho = 0$ if and only if
\begin{eqnarray*}
\log\rho_{ABC} + \log\rho_B = \log\rho_{AB} + \log\rho_{BC}.
\end{eqnarray*}
\end{thrm}

\begin{proof}
Denote
$$
H=\log\rho_{ABC},~~K=\frac12\log\rho_{AB} + \frac12\log\rho_{BC}
-\frac12\log\rho_{ABC} -\frac12\log\rho_B.
$$
Thus $\Tr{e^H}=1$ and $H+K = \frac12\log\rho_{ABC} +
\frac12\log\rho_{AB} + \frac12\log\rho_{BC} -\frac12\log\rho_B$.
Since
$$
I(A:C|B)_\rho = \Tr{\rho_{ABC}(\log\rho_{ABC} + \log\rho_B -
\log\rho_{AB}-\log\rho_{BC})},
$$
it follows from Peierls-Bogoliubov inequality and Golden-Thompson
inequality that
\begin{eqnarray*}
&&\exp\Pa{-\frac12I(A:C|B)_\rho} = \exp\Pa{\Tr{e^HK}} = \exp\Pa{\frac{\Tr{e^HK}}{\Tr{e^H}}}\\
&&\leqslant \frac{\Tr{e^{H+K}}}{\Tr{e^H}} = \Tr{e^{H+K}}\\
&&= \Tr{\exp\Pa{\frac12\log\rho_{ABC} + \frac12\log\rho_{AB} +
\frac12\log\rho_{BC} -\frac12\log\rho_B}}\\
&&\leqslant
\Tr{\exp\Pa{\frac12\log\rho_{ABC}}\exp\Pa{\frac12\log\rho_{AB} +
\frac12\log\rho_{BC} -\frac12\log\rho_B}}\\
&&=\Tr{\sqrt{\rho_{ABC}}\sqrt{\exp\Pa{\log\rho_{AB} + \log\rho_{BC}
-\log\rho_B}}},
\end{eqnarray*}
which implies that
\begin{eqnarray}
I(A:C|B)_\rho \geqslant
-2\log\Tr{\sqrt{\rho_{ABC}}\sqrt{\exp\Pa{\log\rho_{AB} +
\log\rho_{BC} -\log\rho_B}}}.
\end{eqnarray}
It is known that for
any positive semi-definite matrices $X,Y$,
$$
\Tr{\sqrt{X}\sqrt{Y}} = \frac{\Tr{X}+\Tr{Y}-\Tr{(\sqrt{X} -
\sqrt{Y})^2}}2.
$$
From the above formula, we have
\begin{eqnarray*}
&&\Tr{\sqrt{\rho_{ABC}}\sqrt{\exp\Pa{\log\rho_{AB} + \log\rho_{BC}
-\log\rho_B}}} \\
&&= \frac{1+\Tr{\exp\Pa{\log\rho_{AB} + \log\rho_{BC}
-\log\rho_B}}}2 \\
&&~~~- \frac12\Tr{\Pa{\sqrt{\rho_{ABC}} -
\sqrt{\exp\Pa{\log\rho_{AB} + \log\rho_{BC} -\log\rho_B}}}^2}.
\end{eqnarray*}
For any positive definite matrices $R,S,T$, we have \cite{EHLieb}
\begin{eqnarray}
\Tr{\exp\Pa{\log R - \log S + \log T}}\leqslant \Tr{\int^{+\infty}_0
R(S+x\I)^{-1}T(S+x\I)^{-1}dx}.
\end{eqnarray}
Taking $R=\rho_{AB},S=\rho_B$, and $T = \rho_{BC}$ in the above inequality gives rise to
\begin{eqnarray*}
&&\ptr{ABC}{\exp\Pa{\log\rho_{AB} - \log\rho_B + \log\rho_{BC}}}\\
&&\leqslant \Ptr{ABC}{\int^{+\infty}_0
\rho_{AB}(\rho_B+x\I)^{-1}\rho_{BC}(\rho_B+x\I)^{-1}dx}\\
&&=\Ptr{AB}{\int^{+\infty}_0
\rho_{AB}(\rho_B+x\I)^{-1}\rho_B(\rho_B+x\I)^{-1}dx}\\
&&=\Ptr{B}{\int^{+\infty}_0
\rho_B(\rho_B+x\I)^{-1}\rho_B(\rho_B+x\I)^{-1}dx} = \ptr{B}{\rho_B} = 1.
\end{eqnarray*}
This fact indicates that $\Tr{\exp\Pa{\log\rho_{AB} + \log\rho_{BC}
-\log\rho_B}}\leqslant 1$. Hence
\begin{eqnarray*}
&&\Tr{\sqrt{\rho_{ABC}}\sqrt{\exp\Pa{\log\rho_{AB} + \log\rho_{BC}
-\log\rho_B}}} \\
&&\leqslant 1 - \frac12\norm{\sqrt{\rho_{ABC}} -
\sqrt{\exp\Pa{\log\rho_{AB} + \log\rho_{BC} -\log\rho_B}}}^2_2.
\end{eqnarray*}
Now since $-\log(1-t)\geqslant t$ for $t\leqslant 1$, it follows
that
\begin{eqnarray*}
I(A:C|B)_\rho &\geqslant&
-2\log\Tr{\sqrt{\rho_{ABC}}\sqrt{\exp\Pa{\log\rho_{AB} +
\log\rho_{BC} -\log\rho_B}}}\\
&\geqslant& -2\log\Pa{1 - \frac12\norm{\sqrt{\rho_{ABC}} -
\sqrt{\exp\Pa{\log\rho_{AB} + \log\rho_{BC} -\log\rho_B}}}^2_2}\\
&\geqslant& \norm{\sqrt{\rho_{ABC}} - \sqrt{\exp\Pa{\log\rho_{AB} +
\log\rho_{BC} -\log\rho_B}}}^2_2.
\end{eqnarray*}
Therefore the desired inequality is obtained.

Now if the conditional mutual information is vanished, then
$$
\norm{\sqrt{\rho_{ABC}} - \sqrt{\exp\Pa{\log\rho_{AB} +
\log\rho_{BC} -\log\rho_B}}}_2=0,
$$
that is, $\sqrt{\rho_{ABC}} = \sqrt{\exp\Pa{\log\rho_{AB} +
\log\rho_{BC} -\log\rho_B}}$, which is equivalent to the following:
$$
\rho_{ABC} = \exp\Pa{\log\rho_{AB} + \log\rho_{BC} -\log\rho_B}.
$$
By taking logarithm over both sides, it is seen that $\log\rho_{ABC}
= \log\rho_{AB} + \log\rho_{BC} -\log\rho_B$, a well-known equality
condition of strong subadditivity obtained by Ruskai in
\cite{Ruskai2002}. This completes the proof.
\end{proof}

\begin{cor}
It holds that
\begin{eqnarray}
I(A:C|B)_\rho \geqslant\frac14 \norm{\rho_{ABC} -
\exp\Pa{\log\rho_{AB} + \log\rho_{BC} -\log\rho_B}}^2_1.
\end{eqnarray}
\end{cor}

\begin{proof}
There is a well-known inequality in Matrix Analysis, i.e.
\emph{Audenaert's inequality} \cite{Audenaert}:
\begin{eqnarray}
\Tr{M^tN^{1-t}}\geqslant\frac12\Tr{M+N - \abs{M-N}}
\end{eqnarray}
for all $t\in[0,1]$ and positive matrices $M,N$, implying that for $t=\frac12$,
$$
\Tr{\sqrt{M}\sqrt{N}}\geqslant\frac12\Tr{M+N - \abs{M-N}}.
$$
Now that
$$
\Tr{\sqrt{M}\sqrt{N}} = \frac12\Tr{M+N - \Pa{\sqrt{M} - \sqrt{N}}^2}.
$$
Thus
\begin{eqnarray}
\norm{\sqrt{M} - \sqrt{N}}^2_2 \leqslant\norm{M-N}_1.
\end{eqnarray}
This is the famous
\emph{Powers-St\"{o}rmer's inequality} \cite{Powers}. Furthermore,
$$
\norm{M-N}_1 \leqslant\norm{\sqrt{M} - \sqrt{N}}_2 \norm{\sqrt{M} + \sqrt{N}}_2.
$$
Indeed, by triangular inequality and Schwartz inequality, it follows
that
\begin{eqnarray*}
\norm{M-N}_1 &=& \norm{\frac12\Pa{\sqrt{M} - \sqrt{N}}\Pa{\sqrt{M} + \sqrt{N}} + \frac12\Pa{\sqrt{M} + \sqrt{N}}\Pa{\sqrt{M} - \sqrt{N}}}_1\\
&\leqslant& \frac12\norm{\Pa{\sqrt{M} - \sqrt{N}}\Pa{\sqrt{M} + \sqrt{N}}}_1 + \frac12\norm{\Pa{\sqrt{M} + \sqrt{N}}\Pa{\sqrt{M} - \sqrt{N}}}_1\\
&\leqslant& \norm{\sqrt{M} - \sqrt{N}}_2 \norm{\sqrt{M} +
\sqrt{N}}_2.
\end{eqnarray*}
Therefore
\begin{eqnarray}
\norm{\sqrt{M} - \sqrt{N}}^2_2 \leqslant\norm{M-N}_1\leqslant
\norm{\sqrt{M} - \sqrt{N}}_2 \norm{\sqrt{M} + \sqrt{N}}_2.
\end{eqnarray}
It follows that
\begin{eqnarray}\label{eq:1-vs-2-norm}
\frac1{\norm{\sqrt{\rho}+\sqrt{\sigma}}^2_2}\norm{\rho -
\sigma}^2_1\leqslant \norm{\sqrt{\rho} - \sqrt{\sigma}}^2_2\leqslant
\norm{\rho - \sigma}_1.
\end{eqnarray}
In view of the fact that
$\norm{\sqrt{\rho}+\sqrt{\sigma}}_2\in[\sqrt{2},2]$, we have
$$
\frac14\leqslant\frac1{\norm{\sqrt{\rho}+\sqrt{\sigma}}^2_2}\leqslant\frac12.
$$
Applying Eq.~\eqref{eq:1-vs-2-norm} to
Eq.~\eqref{eq:lower-bound} in Theorem~\ref{th:newbound}, we get the
desired inequality.
\end{proof}
The following is the second one of main results:
\begin{thrm}
For two density matrices $\rho,\sigma\in\density{\cH}$ and a quantum
channel $\Phi$ over $\cH$, we have
\begin{eqnarray*}
&&\rS(\rho||\sigma) - \rS(\Phi(\rho)||\Phi(\sigma)) \\
&&\geqslant -2\log\Tr{\sqrt{\rho}\sqrt{\exp\Br{\log\sigma +
\Phi^*(\log\Phi(\rho)) - \Phi^*(\log\Phi(\sigma))}}},
\end{eqnarray*}
where $\Phi^*$ is a dual of $\Phi$ with respect to Hilbert-Schmidt
inner product over the operator space on $\cH$.
\end{thrm}

\begin{proof}
Since
\begin{eqnarray*}
&&\rS(\Phi(\rho)||\Phi(\sigma)) - \rS(\rho||\sigma) \\
&&= \Tr{\rho\Br{-\log\rho + \log\sigma + \Phi^*(\log\Phi(\rho)) -
\Phi^*(\log\Phi(\sigma))}},
\end{eqnarray*}
it follows from Golden-Thompson inequality that
\begin{eqnarray*}
&&\exp\Pa{\frac12\rS(\Phi(\rho)||\Phi(\sigma)) -
\frac12\rS(\rho||\sigma)} \\
&&\leqslant \Tr{\exp\Br{\frac12\log\rho + \frac12\log\sigma +
\frac12\Phi^*(\log\Phi(\rho)) - \frac12\Phi^*(\log\Phi(\sigma))}}\\
&&\leqslant\Tr{\sqrt{\rho}\sqrt{\exp\Br{\log\sigma +
\Phi^*(\log\Phi(\rho)) - \Phi^*(\log\Phi(\sigma))}}},
\end{eqnarray*}
which implies the desired inequality.
\end{proof}

\section{Some remarks}

\begin{remark}
It is clear that Brand\~{a}o \emph{et al}'s bound \eqref{eq:fernando},
and Li and Winter's bound \eqref{eq:kli} are both LOCC
measurement-based. Moreover, they are independent of system $B$, in
view of this, they gave a lower bound of squashed entanglement,
defined by the following \cite{Brandao2011}:
\begin{eqnarray}
E_{sq}(\rho_{AC}) = \inf_{B}\Set{\frac12I(A:C|B)_\rho:
\ptr{B}{\rho_{ABC}} = \rho_{AC}}.
\end{eqnarray}
However, although our result depends on the system $B$, but, it shed
new light over squashed entanglement. More topics related with our
bound can be found in \cite{Leifer2008,Poulin2011,Kim2013}.
\end{remark}

It is asked in \cite{Zhang2013}: Can we derive
$I(A:C|B)_\rho = 0$ from $\Br{M,M^\dagger} = 0$ ? The answer is negative. Indeed, it follows from the discussion in \cite{Winter} that if the
operators $\rho_{AB},\rho_{BC}$ and $\rho_B$ are commute, then
\begin{eqnarray}\label{eq:Pinsker}
I(A:C|B)_\rho = \rS(\rho_{ABC}||MM^\dagger).
\end{eqnarray}
Now let $\rho_{ABC} = \sum_{i,j,k} p_{ijk}\out{ijk}{ijk}$ with
$\set{p_{ijk}}$ being an arbitrary joint probability distribution.
Thus
$$
MM^\dagger = M^\dagger M = \sum_{i,j,k}
\tfrac{p_{ij}p_{jk}}{p_j}\out{ijk}{ijk},
$$
where $p_{ij}=\sum_k p_{ijk},p_{jk}=\sum_ip_{ijk}$ and
$p_j=\sum_{i,k}p_{ijk}$ are corresponding marginal distributions,
respectively. In general,
$p_{ijk}\neq\tfrac{p_{ij}p_{jk}}{p_j}$. Therefore we have a specific
example in which $[M,M^\dagger]=0$, and $\rho_{ABC}\neq MM^\dagger$,
i.e. $I(A:C|B)_\rho>0$. By employing the Pinsker's inequality to
Eq.~\eqref{eq:Pinsker}, in this special case, it follows that
$$
I(A:C|B)_\rho \geqslant \frac12\norm{\rho_{ABC} - MM^\dagger}^2_1.
$$

Along with the above line, all tripartite states can be classified
into three categories:
$$
\density{\cH_A\ot\cH_B\ot\cH_C} = \sD_1\cup \sD_2 \cup \sD_3,
$$
where
\begin{enumerate}[(i)]
\item $\sD_1 \defeq \Set{\rho_{ABC}:\rho_{ABC} = MM^\dagger, [M,M^\dagger]=0}$.
\item $\sD_2 \defeq \Set{\rho_{ABC}:\rho_{ABC} \neq MM^\dagger, [M,M^\dagger]=0}$.
\item $\sD_3 \defeq \Set{\rho_{ABC}:[M,M^\dagger]\neq0}$.
\end{enumerate}
It is remarked here that for any tripartite state $\rho_{ABC}$, a
transformation can be defined as follows:
\begin{eqnarray}
\sM(\rho_{ABC}) :=
\rho^{1/2}_{AB}\rho^{-1/2}_B\rho_{BC}\rho^{-1/2}_B
\rho^{1/2}_{AB}~~\text{for}~~\forall
\rho_{ABC}\in\density{\cH_A\ot\cH_B\ot\cH_C}.
\end{eqnarray}
Apparently $\sM$ is a quantum channel since $\sM =
\Phi^*_\sigma\circ\Phi$ with $\Phi=\trace_A$ and
$\sigma=\rho_{AB}\ot\rho_C$. In general, the output state of $\sM$
is not a \emph{Markov chain state}, that is, the so-called state
with vanishing quantum conditional mutual information, unless
$\rho_{ABC}$ is a Markov chain state. Another analogous
transformation can be defined
\begin{eqnarray}
\sM'(\rho_{ABC}) :=
\rho^{1/2}_{BC}\rho^{-1/2}_B\rho_{AB}\rho^{-1/2}_B
\rho^{1/2}_{BC}~~\text{for}~~\forall
\rho_{ABC}\in\density{\cH_A\ot\cH_B\ot\cH_C}.
\end{eqnarray}
In \cite{Zhang2013}, it is \emph{conjectured} that
\begin{eqnarray}\label{conjecture}
I(A:C|B)_\rho \geqslant \frac12\max\Set{\norm{\rho_{ABC} -
\sM(\rho_{ABC})}^2_1, \norm{\rho_{ABC} - \sM'(\rho_{ABC})}^2_1}.
\end{eqnarray}
In fact, under the condition that
$$
\sM(\rho_{ABC}) = \exp(\log\rho_{AB}-\log\rho_B + \log\rho_{BC}),
$$
it is still seen that
$$
I(A:C|B)_\rho = \rS(\rho_{ABC}||\sM(\rho_{ABC})),
$$
implying (by Pinsker's inequality) that the conjectured inequality
is also true.

We can connect the total amount of conditional mutual information
contained in the tripartite state $\rho_{ABC}$ with the trace-norm
of the commutator $\Br{M,M^\dagger}$ as follows: if the above
conjecture holds, then we have
\begin{eqnarray}\label{eq:conj}
I(A:C|B)_\rho \geqslant \frac18\norm{\Br{M, M^\dagger}}^2_1,
\end{eqnarray}
but not vice versa. Even though the above conjecture is false, it is
still possible that this inequality is true. From the classification
of all tripartite states, it suffices to show Eq.~\eqref{eq:conj} is
true for states in $\sD_3$.

In fact, by using Wasin-So Identity several times, it follows that
\begin{eqnarray}
\sM(\rho_{ABC}) &=& \exp\Pa{U\log\rho_{AB}U^\dagger +
V\log\rho_{BC}V^\dagger -W\log\rho_BW^\dagger},\\
\sM'(\rho_{ABC}) &=& \exp\Pa{U'\log\rho_{AB}U'^\dagger +
V'\log\rho_{BC}V'^\dagger -W'\log\rho_BW'^\dagger}
\end{eqnarray}
for some triples of unitaries $(U,V,W)$ and $(U',V',W')$ over
$\cH_{ABC}$. The following conjecture is left \emph{open}: For all
triple of unitaries $(U,V,W)$ over $\cH_{ABC}$,
\begin{eqnarray}
I(A:C|B)_\rho \geqslant\frac14 \norm{\rho_{ABC} -
\exp\Pa{U\log\rho_{AB}U^\dagger + V\log\rho_{BC}V^\dagger
-W\log\rho_BW^\dagger}}^2_1.
\end{eqnarray}
Once this inequality is proved, a weaker one would be true:
\begin{eqnarray}
I(A:C|B)_\rho \geqslant\frac14 \max\Set{\norm{\rho_{ABC} -
\sM(\rho_{ABC})}^2_1, \norm{\rho_{ABC} - \sM'(\rho_{ABC})}^2_1}.
\end{eqnarray}

\begin{remark}
If one can show that
$$
\Tr{\exp\Pa{\log\sigma + \Phi^*(\log\Phi(\rho)) -
\Phi^*(\log\Phi(\sigma))}}\leqslant1,
$$
then it would be true that
$$
\rS(\rho||\sigma) - \rS(\Phi(\rho)||\Phi(\sigma))\geqslant
\frac14\norm{\rho - \exp\Pa{\log\sigma + \Phi^*(\log\Phi(\rho)) -
\Phi^*(\log\Phi(\sigma))}}^2_1.
$$
By similar reasoning in the previous part, it is believed that
\begin{eqnarray*}
\rS(\rho||\sigma) - \rS(\Phi(\rho)||\Phi(\sigma))\geqslant
\frac14\norm{\rho - \exp\Pa{U\log\sigma U^\dagger +
V\Phi^*(\log\Phi(\rho))V^\dagger -
W\Phi^*(\log\Phi(\sigma))W^\dagger}}^2_1.
\end{eqnarray*}
And
\begin{eqnarray*}
\rS(\rho||\sigma) - \rS(\Phi(\rho)||\Phi(\sigma))\geqslant
\frac14\norm{\rho - \Phi^*_\sigma\circ\Phi(\rho)}^2_1.
\end{eqnarray*}
The crack of this problem amounts to give a solution of Li and
Winter's question \cite{Winter} from a different perspective.
\end{remark}

\begin{remark} By using the Peierls-Bogoliubov inequality and Golden-Thompson inequality, we show the following
interesting inequality: We know from \cite{zhang2014} that there
exists a unitary $U$ (obtained by Golden-Thomspon inequality with
equality condition and Wasin-So Identity) such that
$\rF(\rho,\sigma) = \Tr{\exp\Pa{\log\sqrt{\rho} + U\log\sqrt{\sigma}
U^\dagger}}$. Now the Peierls-Bogoliubov inequality is used to give
a new lower bound for fidelity:
\begin{eqnarray*}
\rF(\rho,\sigma)&\geqslant&\Tr{\sqrt{\rho}}\exp\Pa{\frac{\Tr{\sqrt{\rho}U\log\sqrt{\sigma}U^\dagger}}{\Tr{\sqrt{\rho}}}}
\geqslant\Tr{\sqrt{\rho}}\exp\Pa{\frac{\Inner{\sqrt{\lambda^\downarrow(\rho)}}{\log\sqrt{\lambda^\uparrow(\sigma)}}}{\Tr{\sqrt{\rho}}}}\\
&\geqslant&\Tr{\sqrt{\rho}}\exp\Pa{\Inner{\sqrt{\lambda^\downarrow(\rho)}}{\log\sqrt{\lambda^\uparrow(\sigma)}}}
\geqslant\Tr{\sqrt{\rho}}\prod_{j=1}^n
\Pa{\lambda^\uparrow_j(\sigma)}^{\frac12\sqrt{\lambda^\downarrow_j(\rho)}}\\
&\geqslant&\Tr{\sqrt{\rho}}\sqrt{\prod_{j=1}^n\Pa{\lambda^\uparrow_j(\sigma)}^{\lambda^\downarrow_j(\rho)}}
=\Tr{\sqrt{\rho}}\exp\Pa{-\frac12\rS(\rho) -
\frac12\rH(\lambda^\downarrow(\rho)||\lambda^\uparrow(\sigma))}.
\end{eqnarray*}
for non-singular density matrices $\rho,\sigma$. Therefore, it is
obtained that for non-singular density matrices
$\rho,\sigma\in\density{\cH}$,
\begin{eqnarray}
\rF(\rho,\sigma)\geqslant\Tr{\sqrt{\rho}}\exp\Pa{-\frac12\rS(\rho) -
\frac12\rH(\lambda^\downarrow(\rho)||\lambda^\uparrow(\sigma))},
\end{eqnarray}
where $\lambda^\downarrow(\rho)$ (resp. $\lambda^\uparrow(\rho)$) is
the probability vector consisted of the eigenvalues of $\rho$,
listed in decreasing (resp. increasing) order; $\rH(\cdot||\cdot)$
is the relative entropy between probability distributions (related
notation can be referred to \cite{zhang2014}).
\end{remark}


\subsection*{Acknowledgements}
Fernando Brand\~{a}o, Shunlong Luo, and Mark Wilde is acknowledged
for their insightful comments on the present manuscript. LZ would
like to thank Patrick Hayden for his important remarks. The work is
supported by the National Natural Science Foundation of China (11301124,
11171301) and by the Doctoral Programs Foundation of Ministry of
Education of China (J20130061).




\begin{thebibliography}{99}

\bibitem{Lieb1973}
E. Lieb and M. Ruskai, {\em Proof of the strong subadditivity of
quantum-mechanical entropy}, J. Math. Phys.
\href{http://dx.doi.org/10.1063/1.1666274}{\textbf{14}, 1938-1941
(1973).}

\bibitem{Kim2012}
I. Kim, {\em Operator extension of strong subadditivity of entropy},
J. Math. Phys.
\href{http://dx.doi.org/10.1063/1.4769176}{\textbf{53}, 122204
(2012).}

\bibitem{Ruskai2012}
M. Ruskai, {\em Remarks on on Kim's strong subadditivity matrix
inequality: extensions and equality conditions}, J. Math. Phys.
\href{http://dx.doi.org/10.1063/1.4823581}{\textbf{54}, 102202
(2013).}

\bibitem{Ruskai2002}
M. Ruskai, {\em Inequalities for quantum entropy: A review with
conditions for equality}, J. Math. Phys.
\href{http://dx.doi.org/10.1063/1.1497701}{\textbf{43}, 4358-4375
(2002)}; erratum
\href{http://dx.doi.org/10.1063/1.1824214}{\textbf{46}, 019901
(2005)}.

\bibitem{Petz2003}
D. Petz, {\em Monotonicity of quantum relative entropy revisited},
Rev. Math. Phys.
\href{http://dx.doi.org/10.1142/S0129055X03001576}{\textbf{15},
79-91 (2003).}

\bibitem{Hayden2004}
P. Hayden, R. Jozsa, D. Petz, A. Winter, {\em Structure of states
which satisfy strong subadditivity of quantum entropy with
equality}, Commun. Math. Phys.
\href{http://dx.doi.org/10.1007/s00220-004-1049-z}{\textbf{246},
359-374 (2004).}

\bibitem{Zhang2013}
Lin Zhang, {\em Conditional mutual information and commutator}, Int.
J. Theor. Phys.
\href{http://dx.doi.org/10.1007/s10773-013-1505-7}{\textbf{52}(6):
2112-2117 (2013).}

\bibitem{Brandao2011}
F. Brand\~{a}o, M. Christandl, J. Yard, {\em Faithful squashed
entanglement}, Commun. Math. Phys.
\href{http://dx.doi.org/10.1007/s00220-011-1302-1}{\textbf{306},
805-830 (2011).} {\em Erratum to: Faithful Squashed Entanglement},
\href{http://dx.doi.org/10.1007/s00220-012-1584-y}{\textbf{316},
287-288 (2012).}
\href{http://arxiv.org/abs/1010.1750}{arXiv:1010.1750v5}

\bibitem{Li2014}
K. Li and A. Winter, {\em Relative entropy and squahsed
entanglement}, Commun. Math. Phys.
\href{http://dx.doi.org/10.1007/s00220-013-1871-2}{\textbf{326},
63-80 (2014).}

\bibitem{Ibinson2008}
B. Ibinson, N. Linden, A. Winter, {\em Robustness of Quantum Markov
Chains}, Commun. Math. Phys.
\href{http://dx.doi.org/10.1007/s00220-007-0362-8}{\textbf{277},
289-304 (2008).}

\bibitem{Petz1988}
D. Petz, {\em Sufficiency of channels over von Neumann algebras},
Quart. J. Math. \href{http://dx.doi.org/10.1093/qmath/39.1.97}{\textbf{39} (1), 97-108 (1988)}.

\bibitem{Hiai2011}
F. Hiai, M. Mosonyi, D. Petz, C. B\'{e}ny, {\em Quantum
$f$-divergences and error correction}, Rev. Math. Phys.
\href{http://dx.doi.org/10.1142/S0129055X11004412}{\textbf{23}(7),
691-747 (2011).}

\bibitem{Bebiano}
N. Bebianoa, R. Lemosb, J. da Provid\^{e}ncia, {\em Inequalities for
quantum relative entropy}, Linear Alg. Appl.
\href{http://dx.doi.org/10.1016/j.laa.2004.03.023}{\textbf{401},
159-172 (2005).}

\bibitem{Forrester}
P.J. Forrester, C.J. Thompson, {\em The Golden-Thompson inequality:
Historical aspects and random matrix applications}, J. Math. Phys.
\href{http://dx.doi.org/10.1063/1.4863477}{\textbf{55}, 023503
(2014).}

\bibitem{So1}
Wasin-So, {\em The high road to an exponential formula}, Linear Alg.
Appl.
\href{http://dx.doi.org/10.1016/S0024-3795(02)00738-3}{\textbf{379},
69-75 (2004).}

\bibitem{So2}
Wasin-So, {\em Equality cases in matrix exponential inequalities},
SIAM J. Matrix Anal. Appl. \href{http://dx.doi.org/10.1137/0613070
}{\textbf{13}, 1154-1158 (1992).}

\bibitem{Carlen}
E.A. Carlen and E.H. Lieb, {\em Remainder Terms for Some Quantum
Entropy Inequalities}, J. Math. Phys.
\href{http://dx.doi.org/10.1063/1.4871575}{\textbf{55}, 042201 (2014).}

\bibitem{EHLieb}
E.H. Lieb, {\em Convex trace functions and the Wigner-Yanase-Dyson
conjecture}, Adv. Math.
\href{http://dx.doi.org/10.1016/0001-8708(73)90011-X}{ \textbf{11},
267-288 (1973).}

\bibitem{Audenaert}
K.M.R. Audenaert, M. Nussbaum, A. Szko{\l}a, F. Verstraete, {\em
Asymptotic Error Rates in Quantum Hypothesis Testing}, Commun. Math.
Phys.
\href{http://dx.doi.org/10.1007/s00220-008-0417-5}{\textbf{279},
251-283 (2008).}


\bibitem{Powers}
R.T. Powers, E. St\"{o}rmer, {\em Free states of the canonical
anticommutation relations}, Commun. Math. Phys.
\href{http://dx.doi.org/10.1007/BF01645492 }{\textbf{16}, 1-33
(1970).}

\bibitem{Leifer2008}
M. Leifer, D. Poulin, {\em Quantum graphical models and belief
propagation}, Ann. Phys.
\href{http://dx.doi.org/10.1016/j.aop.2007.10.001}{\textbf{323}(8),
1899-1946 (2008).}

\bibitem{Poulin2011}
D. Poulin, M.B. Hastings, {\em Markov Entropy Decomposition: A
Variational Dual for Quantum Belief Propagation}, Phys. Rev. Lett.
\href{htt://dx.doi.org/10.1103/PhysRevLett.106.080403}{\textbf{106},
080403 (2011).}

\bibitem{Kim2013}
I.H. Kim, {\em Perturbative analysis of topological entanglement
entropy from conditional independence}, Phys. Rev. B
\href{http://dx.doi.org/10.1103/PhysRevB.87.155120}{\textbf{87},
155120 (2013).}


\bibitem{Winter}
A. Winter and K. Li, {\em A stronger subadditivity relation? With
applications to squashed entanglement, sharability and
separability}, see
\url{http://www.maths.bris.ac.uk/~csajw/stronger_subadditivity.pdf}


\bibitem{zhang2014}
L. Zhang, S.-M. Fei, {\em Quantum fidelity and relative entropy between
unitary orbits}, J. Phys. A: Math. Theor.
\href{http://dx.doi.org/10.1088/1751-8113/47/5/055301}{\textbf{47}, 055301
(2014).}


\end{thebibliography}
\end{document}